\documentclass[english,reprint,superscriptaddress, nofootinbib]{revtex4-1}
\usepackage[T1]{fontenc}
\usepackage[latin9]{inputenc}
\usepackage{units}
\usepackage{amsmath}
\usepackage{amsthm}

\makeatletter
\theoremstyle{plain}
\newtheorem*{thm*}{\protect\theoremname}

\usepackage{babel}
\providecommand{\theoremname}{Theorem}

\makeatother

\usepackage{babel}
\providecommand{\theoremname}{Theorem}

\begin{document}
\title{Systems of random variables and the Free Will Theorem}
\author{Ehtibar N.\ Dzhafarov}
\email{To whom correspondence should be addressed. E-mail: ehtibar@purdue.edu}

\affiliation{Purdue University, USA}
\author{Janne V.\ Kujala}
\email{E-mail: jvk@iki.fi}

\affiliation{University of Turku, Finland}
\begin{abstract}
The title refers to the Free Will Theorem by Conway and Kochen whose
flashy formulation is: if experimenters possess free will, then so
do particles. In more modest terms, the theorem says that individual
pairs of spacelike separated particles cannot be described by deterministic
systems provided their mixture is the same for all choices of measurement
settings. We reformulate and generalize the Free Will Theorem in terms
of systems of random variables, and show that the proof is based on
two observations: (1) some compound systems are contextual (non-local),
and (2) any deterministic system with spacelike separated components
is non-signaling. The contradiction between the two is obtained by
showing that a mixture of non-signaling deterministic systems, if
they exist, is always noncontextual. The ``experimenters' free will''
(independence) assumption is not needed for the proof: it is made
redundant by the assumption (1) above, critical for the proof. We
next argue that the reason why an individual pair of particles is
not described by a deterministic system is more elementary than in
the Free Will Theorem. A system, contextual or not and deterministic
or not, includes several choices of settings, each of which can be
factually used without changing the system. An individual pair of
particles can only afford a single realization of random variables
for a single choice of settings. With this conceptualization, the
``free will of experimenters'' cannot be even meaningfully formulated,
and the choice between the determinism and ``free will of particles''
becomes arbitrary and inconsequential. 
\end{abstract}
\maketitle

\section{Informal introduction}

In this paper the issues related to the Free Will Theorem (FWT) \citep{ConwayKochen2006,ConwayKochen2009}
are discussed in terms of random variables. Conway and Kochen in \citep{ConwayKochen2009,ConwayKochen2007}
emphasize that their theorem does not use probabilistic notions. This
seems to plunge our paper in a controversy from the outset \citep{Tumulka2007,Goldsteinetal.2010,ConwayKochen2007}:
our analysis of the FWT would be suspect if we used conceptual means
that are not acceptable in the original formulation of the theorem.
This is not the case however. We use the language of random variables
to describe quantum experiments, involving large numbers of replications
with multitudes of particles. On the level of individual particles
(more specifically, individual pairs of entangled particles), our
focus, the same as in the original FWT, is exclusively on whether
they can be described as systems of deterministic outcomes (which
are, of course, a special case of random variables). Probabilistic
description of quantum experiments is hardly controversial, and we
have a good demonstration of the benefits it offers for the FWT in
\citep{CatorLandsman2014}. Moreover, it is unavoidable. It is known
\citep{CatorLandsman2014,Goldsteinetal.2010} that the FWT can be
found or extracted from much earlier work than \citep{ConwayKochen2006,ConwayKochen2009},
with the Kochen-Specker system used by Conway and Kochen being replaced
with other contextual systems (in fact, any contextual system, as
we will see later). The contextuality of many of these systems is
more saliently probabilisitic than that of the Kochen-Specker system.
This prominently includes the EPR/Bohm system, with a variant of the
FWT being already seen in Bell's pioneering work \citep{Bell1964}.

Our reformulations lead us to critically re-examine and modify the
FWT, although not invalidate it. In particular, we show that the assumption
of the free will of experimenters (or independence assumption, as
many authors prefer to call it \citep{CatorLandsman2014,Esfeld2015,Friedmanetal2019,BarrettGisin2011,Rossetetal.2014,Hall2011})
is not needed in the FWT. This assumption is only needed to ensure
that experimental observations correctly identify the system experimented
on as contextual. It is therefore made unnecessary by another assumption,
critical for the FWT and underivable from the independence assumption
--- that a contextual system with certain properties exists. Furthermore,
we argue that while the question of whether individual particles can
be described by deterministic systems is indeed to be answered negatively,
and while the FWT is indeed one way of demonstrating this, there is
a more elementary reason for this negative answer: the notion of a
system, deterministic or not, is not applicable to an individual pair
of particles to begin with. The latter is a realization of random
variables for a single choice of settings, whereas the notion of a
system involves several mutually exclusive settings, each of which
can be factually and repeatedly used.

Our analysis is based on the Contextuality-by-Default (CbD) theory
(e.g., \citep{DCK2017,DK2017,KDL2015}), but its utilization in this
paper is confined to only two basic principles. The first one is that
each random variable is identified not only by the property it measures
but also by the context (settings) in which it measures this property.
The second principle is that no two random variables recorded in different,
mutually exclusive contexts possess a joint distribution. Moreover,
these principles are only applied to a special subset of systems of
random variables, the compound, ``Alice-Bob''-type systems with
spacelike separation. These systems are non-signaling, and this makes
it unnecessary for us to invoke most of the content of CbD. Because
of this the reader need not be familiar with CbD to understand this
paper.

However, a brief comment may be needed on the two principles just
mentioned. The double-indexation of random variables means that if
Alice chooses a setting $x$ and Bob chooses a setting $y$, their
measurement outcomes (random variables) are represented as, respectively,
$A^{x,y}$ and $B^{x,y}$. And if Bob changes his setting to $y'$
while Alice maintains her setting, her measurement is represented
by another random variable, $A^{x,y'}.$ A reader might erroneously
interpret this as indicating that Bob somehow influences Alice's measurements
despite their spacelike separation (a ``spooky action at a distance'').
This is not the case. The distribution of $A^{x,y'}$ is the same
as that of $A^{x,y}$, so Bob transfers no information to Alice. There
is no ``action.'' The difference between $A^{x,y}$ and $A^{x,y'}$
simply reflects the relational nature of random variables in classical
probability theory. A random variable is a measurable function on
a probability space, and any variable defined on the same space is
jointly distributed with it: their observed realizations are paired.
Therefore, if $A^{x,y}$ and $A^{x,y'}$ were the same random variable,
they would be jointly distributed. But this would mean that realizations
of $A^{x,y}$ and realizations of $A^{x,y'}$ co-occur (and are equal),
while in reality they occur in mutually exclusive contexts.\footnote{Without elaborating (see \citep{DK2017,Dzh2019} for detailed argument),
another way of understanding the contextual labeling is to observe
that if one dropped the second superscript in $A^{x,y}$ and the first
superscript in $B^{x,y}$, the system would have to be noncontextual
(due to the fact that joint distributions of $X,Y$ and of $Y,Z$
imply the joint distribution of $X,Y,Z$). The existence of contextual
systems therefore is a \emph{reductio ad absurdum} proof that contextual
labeling is necessary.}

The scheme of the paper is as follows. In the next section we introduce
formal notions and definitions related to systems of random variables.
In Sections \ref{sec:Free-Will-Theorem} and \ref{sec:Where-is-the}
we present the FWT in the language of such systems. In Section \ref{sec:Systems-versus-isolated}
we show that a systematic use of the language of random variables
makes the FWT unnecessary (though not wrong): the experimenter' free
will (independence) becomes unformulable, and the choice between the
determinism and free will of particles becomes arbitrary and inconsequential.
The concluding section provides a brief summary.

\section{\label{sec:Preliminaries}Preliminaries}

A \emph{compound system} of random variables is an indexed set of
random variables 
\begin{equation}
\mathcal{R}=\left\{ \left(A^{x,y},B^{x,y}\right):\left(x,y\right)\in C\right\} ,\label{eq:system gen}
\end{equation}
where $x$ is a \emph{property} measured by Alice, $y$ is a \emph{property}
measured by Bob, $\left(x,y\right)$ is the \emph{context} in which
the measurements are made, and $C$ is a set of all possible contexts.
Every random variable therefore is identified by the property it measures
and the context in which it measures it. To simplify discussion, we
will assume that all random variables have a finite number of values.
The properties $x,y$ are also referred to as \emph{settings}, although
there is the obvious semantic difference: a setting $x$ designates
the decision and arrangements made to measure property $x$.

Alice and Bob are always assumed to be \emph{spacelike separated}.
Because of this, by special relativity, the system is \emph{non-signaling}:
the distributions of the variables are context-independent, 
\begin{equation}
A^{x,y}\overset{dist}{=}A^{x,y'},\label{eq:nonsignaling1}
\end{equation}
for any $x,y,y'$ such that $\left(x,y\right),\left(x,y'\right)\in C$.
The symbol $\overset{dist}{=}$ indicates equality of distributions.
Analogously, 
\begin{equation}
B^{x,y}\overset{dist}{=}B^{x',y},
\end{equation}
for any $y,x,x'$ such that $\left(x,y\right),\left(x',y\right)\in C$.

One prominent example of a compound system is the EPR/Bohm system
\citep{Bell1964,CHSH1969}, 
\begin{equation}
\begin{array}{c}
\boxed{\mathcal{R}_{EPRB}}\\
\begin{array}{|c|c||c|c|}
\hline A^{1,1} &  & B^{1,1} & \\
\hline A^{1,2} &  &  & B^{1,2}\\
\hline  & A^{2,1} & B^{2,1} & \\
\hline  & A^{2,2} &  & B^{2,2}
\\\hline \end{array}
\end{array}\;,
\end{equation}
with $C=\left\{ x=1,x=2\right\} \times\left\{ y=1,y=2\right\} $.
Another prominent example is the compound version of the Kochen-Specker-Peres
system \citep{Peres1995,KS1967}, 
\begin{equation}
\begin{array}{c}
\boxed{\mathcal{R}_{KSP}}\\
\begin{array}{|c|c|c|c||c|c|c|c|}
\hline A^{1,1} &  & \cdots &  & B^{1,1} &  & \cdots & \\
\hline A^{1,2} &  & \cdots &  &  & B^{1,2} & \cdots & \\
\hline \vdots &  & \vdots\vdots\vdots &  &  & \vdots & \vdots\vdots\vdots & \\
\hline A^{1,33} &  &  &  &  &  & \cdots & B^{1,33}\\
\hline \vdots & \vdots & \vdots\vdots\vdots & \vdots & \vdots & \vdots & \vdots\vdots\vdots & \vdots\\
\hline  &  & \cdots & A^{40,1} & B^{40,1} &  & \cdots & \\
\hline  &  & \cdots & A^{40,2} &  & B^{40,2} & \cdots & \\
\hline  &  & \vdots\vdots\vdots & \vdots &  & \vdots & \vdots\vdots\vdots & \\
\hline  &  &  & A^{40,33} &  &  & \cdots & B^{40,33}
\\\hline \end{array}
\end{array}\;,
\end{equation}
with $C=\left\{ x=1,\ldots,x=40\right\} \times\left\{ y=1,\ldots,y=33\right\} $.
In $\mathcal{R}_{EPRB}$, the $x$-values and $y$-values enumerate
choices of axes by Alice and Bob, and the random variables are $0/1$
(say, spin values in spin-$\nicefrac{1}{2}$ particles). In $\mathcal{R}_{KSP}$,
the $y$-values represent 33 special axes in Peres's proof of the
Kochen-Specker theorem \citep{Peres1995}, and the $x$-values encode
40 Peres's triples formed using these 33 axes; the $A$-variables
have values $011,101,110$, and the $B$-variables are $0/1$.

Any two random variables recorded in the same context, and referred
to as an \emph{AB-pair}, are \emph{jointly distributed}: this means
that $\Pr\left[A^{x,y}=a,B^{x,y}=b\right]$ is well defined, for $\left(x,y\right)\in C$.
However, two random variables from different contexts are \emph{stochastically
unrelated}, i.e. have no joint distribution: i.e., if $\left(x,y\right)\not=\left(x',y'\right)$,
the event conjunctions $\left[A^{x,y}=a,A^{x',y'}=a'\right]$, $\left[A^{x,y}=a,B^{x',y'}=b\right]$,
etc. are not well-defined events, and no probabilities can be assigned
to them. This formal distinction reflects the obvious fact that random
variables from mutually exclusive contexts can never be observed together,
in any empirical meaning of ``together.'' In particular, $A^{x,y}$
and $A^{x,y'}$ in (\ref{eq:nonsignaling1}) are not equal, because
they are not jointly distributed. All this means that the system $\mathcal{R}$
in (\ref{eq:system gen}) is a collection of stochastically unrelated
AB-pairs, combined within a single system only because every AB-pair
shares at least one property it measures with at least one other AB-pair.

Random variables attaining a given value with probability 1 are \emph{deterministic
variables}. A \emph{deterministic system} is a system containing only
deterministic variables. Thus, the two systems below are deterministic
versions of the EPR/Bohm system, non-signaling ($\mathcal{D}_{EPRB}$)
and signaling ($\mathcal{D}'_{EPRB}$): 
\begin{equation}
\begin{array}{c}
\boxed{\mathcal{D}_{EPRB}}\\
\begin{array}{|c|c||c|c|}
\hline 1 &  & 0 & \\
\hline 1 &  &  & 1\\
\hline  & 0 & 0 & \\
\hline  & 0 &  & 1\\
\hline\hline _{x=1} & _{x=2} & _{y=1} & _{y=2}
\end{array}
\end{array}\;,\begin{array}{c}
\boxed{\mathcal{D}'_{EPRB}}\\
\begin{array}{|c|c||c|c|}
\hline 1 &  & 0 & \\
\hline 0 &  &  & 1\\
\hline  & 0 & 1 & \\
\hline  & 0 &  & 1\\
\hline\hline _{x=1} & _{x=2} & _{y=1} & _{y=2}
\end{array}
\end{array}\;.\label{eq:EPRB D D'}
\end{equation}
In presenting these deterministic systems we conveniently identify
the random variables with their supports, say, writing $0$ instead
of $A^{2,1}\equiv0$ ($\equiv$ meaning ``equal with probability
1''). Because one loses the indexation as a result, one has to indicate
for each number what properties it measures (at the bottom of the
tables).

Note that we use capital Roman letters to designate random variables,
and the script letters $\mathcal{R},\mathcal{D}$ to refer to systems
--- because a system is not a random variable, it is a set of stochastically
unrelated random variables (the AB-pairs).

A \emph{coupling} $\bar{R}$ for a system $\mathcal{R}$ is an identically
double-labeled set of \emph{jointly distributed} random variables
\begin{equation}
\bar{R}\overset{}{=}\left\{ \bar{A}^{x,y},\bar{B}^{x,y}:\left(x,y\right)\in C\right\} ,\label{eq:coupling gen}
\end{equation}
such that every AB-pair of $\bar{R}$ is distributed as the corresponding
AB-pair of $\mathcal{R}$: 
\begin{equation}
\left(\bar{A}^{x,y},\bar{B}^{x,y}\right)\overset{dist}{=}\left(A^{x,y},B^{x,y}\right),
\end{equation}
for every $\left(x,y\right)\in C$. Note that we can write $\bar{R}$
rather than $\mathcal{\bar{R}}$ in (\ref{eq:coupling gen}) because
a coupling is a random variable in its own right. Thus, while a coupling
of $\mathcal{R}_{EPRB}$ can be presented as 
\begin{equation}
\begin{array}{c}
\boxed{\bar{R}_{EPRB}}\\
\begin{array}{|c|c||c|c|}
\hline \bar{A}^{1,1} &  & \bar{B}^{1,1} & \\
\hline \bar{A}^{1,2} &  &  & \bar{B}^{1,2}\\
\hline  & \bar{A}^{2,1} & \bar{B}^{2,1} & \\
\hline  & \bar{A}^{2,2} &  & \bar{B}^{2,2}
\\\hline \end{array}
\end{array}\;,
\end{equation}
it is no longer a set of four stochastically unrelated pairs, but
a random variable 
\begin{equation}
\bar{R}=\left\{ \bar{A}^{1,1},\bar{A}^{1,2},\bar{A}^{2,1},\bar{A}^{2,2},\bar{B}^{1,1},\bar{B}^{2,1},\bar{B}^{1,2},\bar{B}^{2,2}\right\} 
\end{equation}
with $2^{8}$ possible values.\footnote{To prevent misunderstanding, the term ``random variable'' is understood
here in the standard meaning of modern probability theory, with no
restrictions on the codomain set of values: random vectors and random
sets therefore are also random variables.} Note also, that any deterministic system $\mathcal{D}$ has a unique
coupling $D$, and the two are easy to confuse if one uses our convenient
identification of deterministic random variables with their supports.
Thus, the coupling of $\mathcal{D}{}_{EPRB}$ in (\ref{eq:EPRB D D'})
is written precisely as $\mathcal{D}{}_{EPRB}$ itself: 
\begin{equation}
\begin{array}{c}
\boxed{\bar{\mathcal{D}}_{EPRB}}\\
\begin{array}{|c|c||c|c|}
\hline 1 &  & 0 & \\
\hline 1 &  &  & 1\\
\hline  & 0 & 0 & \\
\hline  & 0 &  & 1\\
\hline\hline _{x=1} & _{x=2} & _{y=1} & _{y=2}
\end{array}
\end{array}
\end{equation}

A (non-signaling) compound system $\mathcal{R}$ is \emph{noncontextual}
if it has a coupling $\bar{R}$ such that 
\begin{equation}
\begin{array}{c}
\Pr\left[\bar{A}^{x,y}=\bar{A}^{x,y'}\right]=1,\\
\textnormal{and}\\
\Pr\left[\bar{B}^{x,y}=\bar{B}^{x',y}\right]=1,
\end{array}\label{eq:identity coupling}
\end{equation}
whenever the indicated contexts are defined (belong to $C$). If such
a coupling does not exist, the system is \emph{contextual}.

Overlooking logical subtleties \citep{Dzh2019}, this definition is
equivalent to the traditional definitions of contextuality and locality,
in terms of the non-existence of joint distributions for single-indexed
random variables \citep{Fine1982} and in terms of hidden variables
with noncontextual/local mapping into observables \citep{Bell1964,KS1967}.
Perhaps this becomes more clear on observing that noncontextuality
is equivalent to the existence of a set of jointly distributed \emph{single-indexed}
random variables 
\begin{equation}
\tilde{R}\overset{}{=}\left\{ \tilde{A}^{x},\tilde{B}^{y}:\left(x,y\right)\in C\right\} ,
\end{equation}
such that 
\begin{equation}
\left(\tilde{A}^{x},\tilde{B}^{y}\right)\overset{dist}{=}\left(A^{x,y},B^{x,y}\right),
\end{equation}
for every $\left(x,y\right)\in C$.

\section{\label{sec:Free-Will-Theorem}Free Will Theorem}

If all AB-pairs in a compound system $\mathcal{R}$ are set to specific
values, the resulting system is called a \emph{realization} of $\mathcal{R}$.
The reason this has to be presented as a definition is that, unlike
a realization of a random variable (e.g., an AB-pair), a realization
of a system is not an observable outcome of any experiment: it is
a pure mathematical abstraction, as the variables do not co-occur
across contexts. Recall that each random variable in $\mathcal{R}$
is finite-valued, because of which the set of possible realizations
of $\mathcal{R}$ is finite. Thus, the system $\mathcal{R}_{EPRB}$
has $4^{4}$ realizations, whereas $\mathcal{R}_{KSP}$ has $6^{40\cdot33}$
realizations.

The question posed in the FWT theorem can be formulated thus: given
that an idealized experiment involving an unlimited number of particle
pairs is described by a system $\mathcal{R}$, is it possible that
each individual pair of particles is a deterministic system that coincides
with one of the realizations of $\mathcal{R}$? The crux of the issue
here is in whether the realizations of the systems can describe individual
particle pairs. If not for this constraint, it would be innocuous
(although still objectionable, as we will see later) to say that the
system $\mathcal{R}$ is presentable as a mixture of some of its realizations
$\mathcal{D}_{1},\ldots,\mathcal{D}_{k}$ taken as deterministic systems:
\begin{equation}
\mathcal{R}\overset{dist}{=}\left\{ \begin{array}{ccc}
\mathcal{D}_{1} & \textnormal{with probability } & p_{1}\\
\vdots & \vdots & \vdots\\
\mathcal{D}_{k} & \textnormal{with probability } & p_{k}
\end{array}\right.,\sum_{i=1}^{k}p_{i}=1.\label{eq:decomposition}
\end{equation}
However, in the question asked by the FWT, these deterministic systems
are assumed to describe real physical entities (particle pairs), because
of which they should be physically realizable. In particular, they
are subject to special relativity, and have to be non-signaling. For
instance, in the decomposition of $\mathcal{R}_{EPRB}$, the signaling
deterministic system $\mathcal{D}'_{EPRB}$ in (\ref{eq:EPRB D D'})
is not allowed. 
\begin{thm*}[reformulated and generalized FWT]
A contextual system $\mathcal{R}$ cannot be decomposed as in (\ref{eq:decomposition}),
where $\mathcal{D}_{1},\ldots,\mathcal{D}_{k}$ are non-signaling
deterministic systems each of which coincides with a realization of
$\mathcal{R}$. 
\end{thm*}
\begin{proof}
If $k=0$ (no non-signaling realizations of $\mathcal{R}$ exist),
the theorem is proved. Assume that (\ref{eq:decomposition}) holds
with $k>0$. Introduce a random variable $\Lambda$ such that $\Pr\left[\Lambda=i\right]=p_{i}$
($i=1,\ldots,k$). Each $\mathcal{D}_{i}$, being deterministic, has
a unique coupling $D_{i}$, and the mixture 
\begin{equation}
\bar{R}=\left\{ \begin{array}{ccc}
D_{1} & \textnormal{if } & \Lambda=1\\
\vdots & \vdots & \vdots\\
D_{k} & \textnormal{if } & \Lambda=k
\end{array}\right.,\sum_{i=1}^{k}p_{i}=1\label{eq:decomposition coupling}
\end{equation}
is a coupling of $\mathcal{R}$. In this coupling, 
\begin{equation}
\bar{A}^{x,y}=\left\{ \begin{array}{ccc}
a_{1}^{x,y} & \textnormal{if } & \Lambda=1\\
\vdots & \vdots & \vdots\\
a_{k}^{x,y} & \textnormal{if } & \Lambda=k
\end{array}\right.,
\end{equation}
where $a_{i}^{x,y}$ is the realization of $A^{x,y}$ in system $\mathcal{D}_{i}$.
But $a_{i}^{x,y}=a_{i}^{x,y'}$ for any $\left(x,y\right),\left(x,y'\right)\in C$
(non-signaling). We have then 
\begin{equation}
\bar{A}^{x,y}=\left\{ \begin{array}{ccc}
a_{1}^{x,y'} & \textnormal{if } & \Lambda=1\\
\vdots & \vdots & \vdots\\
a_{k}^{x,y'} & \textnormal{if } & \Lambda=k
\end{array}\right.=\bar{A}^{x,y'}.
\end{equation}
Analogously, 
\begin{equation}
\bar{B}^{x,y}=\left\{ \begin{array}{ccc}
b_{1}^{x,y}=b_{1}^{x',y} & \textnormal{if } & \Lambda=1\\
\vdots & \vdots & \vdots\\
b_{k}^{x,y}=b_{k}^{x',y} & \textnormal{if } & \Lambda=k
\end{array}\right.=\bar{B}^{x',y},
\end{equation}
for all $\left(x,y\right),\left(x',y\right)\in C$. By definition
then, $\mathcal{R}$ is noncontextual, contrary to the theorem's premise. 
\end{proof}
Equivalently, and perhaps more familiar to physicists, the proof could
be formulated as a demonstration that $\mathcal{R}$ has a local hidden
variable model. Using the same $\Lambda$ as in the proof, we have
a coupling $\bar{R}$ of $\mathcal{R}$ such that 
\begin{equation}
\left(\bar{A}^{x,y},\bar{B}^{x,y}\right)=\left(f\left(\Lambda,x,y\right),g\left(\Lambda,x,y\right)\right),
\end{equation}
where for each value $\Lambda=i$ and each $\left(x,y\right)$, the
function $\left(f,g\right)$ reads the value of $\left(a_{i}^{x,y},b_{i}^{x,y}\right)$
in $D_{i}$. Since each $D_{i}$ is non-signaling, i.e., $f\left(i,x,y\right)=f\left(i,x\right)$
and $g\left(i,x,y\right)=g\left(i,y\right)$ for each $\Lambda=i$,
we have 
\begin{equation}
\left(\bar{A}^{x,y},\bar{B}^{x,y}\right)=\left(f\left(\Lambda,x\right),g\left(\Lambda,y\right)\right),
\end{equation}
which is a local (i.e., noncontextual) model with $\Lambda$ as a
hidden variable.

Applying this theorem to $\mathcal{R}_{KSP}$ used in Conway and Kochen's
proof, this system has no non-signaling realizations (by the Kochen-Specker
theorem). Rather surprisingly therefore, the proof of the Conway-Kochen
version of the FWT is contained in the first sentence of the proof
above. For $\mathcal{R}_{EPRB}$, we have 16 non-signaling realizations
of this system, and the proof says that their mixtures can only be
noncontextual.

\section{\label{sec:Where-is-the}Where is the free will assumption in the
proof?}

The formulations and proofs given in the previous section do not even
mention the hypothetical freedom with which Alice and Bob choose their
settings. How is it possible? The answer is that the experimenters'
free will assumption is indeed redundant. The proof above is contingent
on the assumption that $\mathcal{R}$ is a correct description of
an idealized experiment involving an unlimited number of particle
pairs, those whereof we ask whether they could be deterministic systems.
The experimenters' free will is only needed to dismiss a conspiracy
of nature leading to an incorrect identification of the system in
such an experiment. Let us explain this using the system $\mathcal{R}_{EPRB}$.

Consider the possibility that in an EPR/Bohm experiment only four
types of entangled particle pairs are possible, described by the deterministic
systems 
\begin{equation}
\begin{array}{cc}
\begin{array}{c}
\boxed{\mathcal{D}_{1}}\\
\begin{array}{|c|c||c|c|}
\hline 1 &  & 1 & \\
\hline 1 &  &  & 1\\
\hline  & 1 & 1 & \\
\hline  & 1 &  & 1\\
\hline\hline _{x=1} & _{x=2} & _{y=1} & _{y=2}
\end{array}
\end{array} & \begin{array}{c}
\boxed{\mathcal{D}_{2}}\\
\begin{array}{|c|c||c|c|}
\hline 0 &  & 0 & \\
\hline 0 &  &  & 0\\
\hline  & 0 & 0 & \\
\hline  & 0 &  & 0\\
\hline\hline _{x=1} & _{x=2} & _{y=1} & _{y=2}
\end{array}
\end{array}\\
\\
\begin{array}{c}
\boxed{\mathcal{D}_{3}}\\
\begin{array}{|c|c||c|c|}
\hline 0 &  & 0 & \\
\hline 0 &  &  & 1\\
\hline  & 1 & 0 & \\
\hline  & 1 &  & 1\\
\hline\hline _{x=1} & _{x=2} & _{y=1} & _{y=2}
\end{array}
\end{array} & \begin{array}{c}
\boxed{\mathcal{D}_{4}}\\
\begin{array}{|c|c||c|c|}
\hline 1 &  & 0 & \\
\hline 1 &  &  & 0\\
\hline  & 0 & 0 & \\
\hline  & 0 &  & 0\\
\hline\hline _{x=1} & _{x=2} & _{y=1} & _{y=2}
\end{array}
\end{array}
\end{array}
\end{equation}
Then the true system $\mathcal{R}_{EPRB}$ obtained from any mixture
of these four systems is noncontextual. Suppose, however, that whenever
Alice and Bob choose $\left(x,y\right)=\left(1,1\right)$ or $\left(2,1\right)$
(the first and third rows in the matrices), the nature chooses to
supply a pair of particles described either by $\mathcal{D}_{1}$
or by $\mathcal{D}_{2}$, equiprobably; whereas for $\left(x,y\right)=\left(1,2\right)$
and $\left(2,2\right)$ (the second and fourth rows in the matrices)
the nature chooses between $\mathcal{D}_{3}$ and $\mathcal{D}_{4}$
equiprobably. Neither Alice nor Bob nor anyone analyzing their experimental
data has any way of knowing this. Following many replications, the
results will be a statistical estimate of a system in which all random
variables are uniformly distributed, 
\begin{equation}
\left\langle A^{x,y}\right\rangle =\left\langle B^{x,y}\right\rangle =\frac{1}{2},
\end{equation}
and 
\begin{equation}
\begin{array}{c}
\left\langle A^{1,1}B^{1,1}\right\rangle =\left\langle A^{2,1}B^{2,1}\right\rangle =\left\langle A^{2,2}B^{2,2}\right\rangle =\frac{1}{2},\\
\\
\left\langle A^{1,2}B^{1,2}\right\rangle =0.
\end{array}
\end{equation}
This is a non-signaling contextual system (a PR box, \citep{PR1994}),
and if used in the proof of the FWT in place of what we assumed to
be the true system, it will lead one to the wrong conclusion that
no decomposition (\ref{eq:decomposition}) is possible.

To avoid such conspiratorial scenarios one can postulate that the
distribution of the non-signaling deterministic systems in (\ref{eq:decomposition})
is the same for all contexts (all choices of settings). This can be
interpreted in terms of Alice's and Bob's free will, but does not
have to. It would be better therefore to call this assumption \emph{unbiasedness},
but not to multiply terminology we follow the authors who call it
(measurement or setting)\emph{ independence}. The assumption, of course,
is consistent with the choices of settings being perfectly predetermined,
but simply uncorrelated with the occurrences of the different types
of deterministic systems. Moreover, Alice's and Bob's choices may
very well be correlated, it makes no difference.

The example just given, of a noncontextual system being mistaken for
a contextual one, suggests the logical possibility of taking the decomposition
(\ref{eq:decomposition}) for granted, and accounting for the (apparent)
contextuality of the observed system $\mathcal{R}$ either by relaxing
the requirement of non-signaling of the deterministic systems $\mathcal{D}_{i}$
or by exploring deviations from the independence assumption. There
is an obvious reciprocity between the two, and it has indeed been
researched and quantified \citep{Friedmanetal2019,BarrettGisin2011,Rossetetal.2014,Hall2011}.
This line of study is outside the scope of our paper. In the FWT we
are only interested in whether individual particle pairs can be described
by deterministic non-signaling systems, and the answer given is negative.

As the above reasoning shows, the independence assumption is not a
necessary premise of this theorem, because it is obviated by the critical
assumption that certain contextual systems exist. The situation is
this: 
\begin{description}
\item [{(i)}] if we do not assume, e.g., that $\mathcal{R}_{EPRB}$ for
certain quadruples of axes is contextual, then the FWT for this system
cannot be proved whether or not one adopts the independence assumption;
and 
\item [{(ii)}] if we do assume the contextuality of $\mathcal{R}_{EPRB}$
(presumably because we believe experiments or quantum-mechanical theory),
the proof can be carried out without mentioning the independence assumption. 
\end{description}
Analogously, Conway and Kochen have to postulate the existence of
a system $\mathcal{R}_{KSP}$ with certain properties (the SPIN assumption),
ensuring that no consistent assignment of values to Alice's measurements
is possible (the Kochen-Specker theorem); and they also have to postulate
that if Bob's axis coincides with one of the three axes chosen by
Alice, then the corresponding measurement outcomes always coincide
(the TWIN assumption). With these postulates, however, the independence
(part of their MIN assumption) is not needed.\footnote{As noted by critics of the FWT \citep{Tumulka2007,Goldsteinetal.2010},
MIN is not a rigorous statement. It speaks of Alice and Bob's choices
as being made ``freely'' and ``independently,'' of which the former
we translate into the independence assumption. The ``independently''
of Conway and Kochen apparently means that Alice's and Bob's choices
are combined in all possible ways. This assumption is not needed either.
Thus, it is easy to see that in the system $\mathcal{R}_{KSP}$ one
can delete all rows in which Bob's axis is not one of the three axes
chosen by Alice. The assumption in question then will be violated,
but the proof of the (original) FWT will not change at all. (MIN also
includes a statement that can be interpreted as non-signaling principle.)}

Summarizing, the independence (or ``experimenters' free will'')
is only needed if we consider the epistemological question: how can
one be sure that the compound system estimated from an experiment
is truly contextual? The FWT is a conditional statement: if we have
a contextual compound system with certain properties, it cannot be
decomposed as in (\ref{eq:decomposition}).

\section{\label{sec:Systems-versus-isolated}Systems versus isolated AB-pairs}

There is, however, a simpler reason not to use deterministic systems
when describing individual particle pairs. The reason is that a single
pair of particles is a \emph{realization of an isolated AB-pair} of
a system rather than a realization of an entire system. The difference
is that the AB-pair is determined by the factual context chosen by
Alice and Bob, while any given realization of a system also includes
the counterfactual contexts that Alice and Bob ``could have chosen.''
Thus, experimental trials for $\mathcal{R}_{EPRB}$ produce a series
of outcomes, such as 
\begin{equation}
\begin{array}{c|c|c|c}
\vdots & \vdots & \vdots & \vdots\\
0 &  &  & 1\\
 & 0 & 1\\
 & 1 & 1\\
0 &  & 0\\
 & 0 &  & 1\\
\vdots & \vdots & \vdots & \vdots\\
\hline\hline _{x=1} & _{x=2} & _{y=1} & _{y=2}
\end{array}\;.\label{eq:sequence}
\end{equation}
We see no logical reason to think that the observed value $\left(A^{1,2},B^{1,2}\right)=\left(0,1\right)$
is somehow related to any specific values of the AB-pairs for contexts
other than the factual $\left(x,y\right)=\left(1,2\right)$. Such
a relation would only be reasonable if there existed a way to observe
these alternative AB-pairs by factually performing the measurements
for $\left(x,y\right)=\left(1,1\right)$, $\left(2,1\right)$, and
$\left(2,2\right)$ in addition to $\left(x,y\right)=\left(1,2\right)$
on the same pair of particles. This is, however, impossible, and not
because a projective measurement is known to destroy the state of
entanglement. We are not allowed to use quantum mechanical considerations
when conceptually testing the standard quantum mechanical view of
elementary particles.\footnote{In fact, it has recently been established \citep{TavakoliCabello2018,Folettoetal.2020}
that a sequence of POVM-represented measurements (each of which depends
on the settings and outcomes of the previous measurements) can be
performed on the same pair of entangled particles without affecting
the entanglement state.} The reason we cannot speak of the ``the same'' pair of entangled
particles being measured repeatedly using variable settings is logical
rather than physical. It is critical for our analysis that all particle
pairs be generated and prepared in precisely the same way, and their
spin values, for any given choice of settings, be measured in precisely
the same way. A pair of particles after a measurement has been performed
on it simply is not the same pair as it was before, and it should
be relabeled accordingly. Barring such relabeling, a sequence like
(\ref{eq:sequence}) should be treated as several (here, four) unrelated
to each other sequences, each defined by a specific context.

One may now ask seemingly the same question as in the FWT but applied
to pairs of particles treated as realizations of specific AB-pairs:
can they be considered deterministic variables? In other words, the
question is whether each pair of particles can be described as 
\begin{equation}
\begin{array}{r}
\begin{array}{|c||c|}
\hline A^{x,y}\equiv a^{x,y} & B^{x,y}\equiv b^{x,y}\\\hline \end{array}\end{array}\;,
\end{equation}
with the specific settings $(x,y)$ under which the measurements are
recorded. This is, however, a very different question, and the answer
is: yes, if one so wishes, but this makes no difference. Flips of
a fair coin can always be considered a mixture of two deterministic
variables with respective values Head and Tail, each occurring with
probability $\nicefrac{1}{2}$. Following the prevailing tradition
in statistics, a sequence of realizations of an AB-pair can be treated
as a realization of a \emph{random sample} (a set of identically distributed
independent random variables). However, it can also be treated as
a realization of a set of different random variables (e.g., deterministic
ones), randomly alternating. Let us again use $\mathcal{R}_{EPRB}$
for an illustration. Assume that one of its AB-pairs is distributed
as 
\begin{equation}
\begin{array}{c|c|c|c|c}
\text{\textnormal{value}} & \left(1,1\right) & \left(1,0\right) & \left(0,1\right) & \left(0,0\right)\\
\hline \textnormal{probability} & p & 1-p & 0 & 0
\end{array}.
\end{equation}
Clearly, this variable is indistinguishable from any mixture 
\begin{equation}
\left\{ \begin{array}{ccc}
X & \textnormal{with probability } & q\\
Y & \textnormal{with probability } & 1-q
\end{array}\right.
\end{equation}
of the random variables $X,Y$ with the respective distributions 
\begin{equation}
\begin{array}{ccc}
X:\begin{array}{|c|c}
\left(1,1\right) & \left(1,0\right)\\
\hline p_{1} & 1-p_{1}
\end{array}, &  & Y:\begin{array}{|c|c}
\left(1,1\right) & \left(1,0\right)\\
\hline p_{2} & 1-p_{2}
\end{array},\end{array}
\end{equation}
provided 
\begin{equation}
qp_{1}+\left(1-q\right)p_{2}=p.
\end{equation}
In particular (and trivially), $X,Y$ can be viewed as deterministic
variables, 
\begin{equation}
\begin{array}{ccc}
X:\begin{array}{|c|c}
\left(1,1\right) & \left(1,0\right)\\
\hline 1 & 0
\end{array}, &  & Y:\begin{array}{|c|c}
\left(1,1\right) & \left(1,0\right)\\
\hline 0 & 1
\end{array},\end{array}
\end{equation}
mixed as 
\begin{equation}
\left\{ \begin{array}{ccc}
X & \textnormal{with probability } & p\\
Y & \textnormal{with probability } & 1-p
\end{array}\right.\;.
\end{equation}
We see that the question of whether the individual pairs of particles
have ``free will'' of their own, i.e., whether they are deterministic
entities, looses its meaning.

It should also be noted that once this context-wise view of the individual
particle pairs is adopted, one need not be concerned with the independence
assumption (``free will'' of Alice and Bob). In fact this assumption
now is unformulable. Each AB-pair corresponds to a fixed context,
and one cannot informatively say that the AB-pair in a given context
does not depend on context.

\section{\label{sec:Conclusion}Conclusion}

To summarize, the view that follows from the conceptual framework
of CbD does not invalidate the FWT-type theorems, but makes them unnecessary.
These theorems can be viewed as reductio ad absurdum demonstrations
that individual pairs of particles cannot be viewed as deterministic
systems. CbD allows one to streamline these demonstrations, ridding
them of unnecessary assumptions. Moreover, CbD leads one to the view
that the individual pairs of particles should not be viewed as systems
to begin with, only as realizations of random variables for a given
choice of settings. As Peres famously put it, ``unperformed experiments
have no results'' \citep{Peres1978} --- and these non-existing
results should not be appended to factual results, lest one runs into
a contradiction. The reason we can speak of $\mathcal{R}$ in (\ref{eq:system gen})
as a system is that as we switch from one context to another, all
other macroscopic circumstances of measurements (overall experimental
set-up and preparation procedure) remain the same. We can repeatedly
experiment with such a system without changing its defining parameters.
In particular, a deterministic system allows one to repeatedly experiment
with it and obtain the same results for any given context. Individual
pairs of particles, unless they change their identity, afford only
one pair of measurements, in one particular context.

\subsubsection*{Acknowledgements.}

This research was supported by the Purdue University's Research Refresh
award. We are grateful to Pawe{\l } Kurzy{\'{n}}ski for critically
commenting on an earlier draft of the paper. We are also grateful
to Ad\'an Cabello for discussing with us the issue of repeated measurements
on individual particles.

\end{document}